%% file: short_final.tex
\documentclass[letterpaper]{article}
\usepackage{spconf}
\usepackage[T1]{fontenc}
\usepackage[utf8]{inputenc}
\usepackage{amsmath,amssymb,amsthm,booktabs,graphicx,tikz}
\usepackage{cite,balance}
\usepackage{microtype}
\usepackage{pgfplots}
\pgfplotsset{compat=1.16}
\usepgfplotslibrary{groupplots}
\usepackage{orcidlink}
\newcommand{\etal}{\textit{et~al.}}
\usepackage{xurl}
\usepackage{pifont}
\usepackage{xcolor}
\definecolor{okgreen}{RGB}{0,128,0}
\definecolor{nored}{RGB}{178,34,34}
\newcommand{\cmark}{\textcolor{okgreen}{\ding{51}}}
\newcommand{\xmark}{\textcolor{nored}{\ding{55}}}

\usetikzlibrary{arrows.meta,positioning,calc}
\hypersetup{
  pdftitle={Exact Factorisation and Fast GPU Computation of Invertible Constant-Q Transforms},
  pdfauthor={Facundo Franchino and Eloi Moliner and Vesa V\"alim\"aki}
}
\newcommand{\R}{\mathbb R}
\newcommand{\C}{\mathbb C}
\newcommand{\diag}{\operatorname{diag}}
\newcommand{\Rea}{\operatorname{Re}}
\newcommand{\norm}[1]{\lVert#1\rVert}
\newtheorem{theorem}{Theorem}
\title{Exact Factorisation and Fast Computation\\of Invertible Constant-Q Transforms}
\name{Facundo Franchino$^1$ \orcidlink{0009-0008-6654-0401} \qquad Eloi Moliner$^2$ \orcidlink{0000-0001-5719-326X} \qquad Vesa V\"alim\"aki$^2$ \orcidlink{0000-0002-7869-292X}}
\address{$^1$Massachusetts Institute of Technology, Cambridge, MA, USA\\
$^2$Acoustics Lab, 
Dept. of Information and Communications Eng., 
Aalto University, Espoo, Finland}

\usepackage[nolist]{acronym}
\begin{acronym}
\acro{stft}[STFT]{short-time Fourier transform}
\acro{istft}[iSTFT]{Inverse Short-Time Fourier Transform}
\acro{tf}[T-F]{time-frequency}
\acro{fft}[FFT]{fast Fourier transform}
\acro{dft}[DFT]{discrete Fourier transform}
\acro{tfrs}[TFRs]{Time-Frequency Representations}
\acro{cqt}[CQT]{constant-$Q$ transform}
\acro{ldm}[LDM]{Latent Diffusion Model}
\acro{snr}[SNR]{Signal-to-Noise Ratio}
\acro{nsgt}[NSGT]{nonstationary Gabor transform}
\acro{erb}[ERB]{Equivalent Rectangular Bandwidth}
\end{acronym}

\makeatletter
\g@addto@macro\small{%
  \setlength\abovedisplayskip{4pt plus 2pt minus 1pt}%
  \setlength\belowdisplayskip{4pt plus 2pt minus 1pt}%
  \setlength\abovedisplayshortskip{0pt plus 2pt}%
  \setlength\belowdisplayshortskip{2pt plus 2pt minus 1pt}}
\makeatother
\begin{document}
\frenchspacing
\ninept
\raggedbottom
\maketitle
\pagestyle{empty}
\thispagestyle{empty}

\begin{abstract}
The constant-$Q$ transform (CQT) represents audio on a logarithmic frequency axis.
Its nonstationary Gabor formulation is exactly invertible, but the unequal numbers of time coefficients in its bands complicate GPU computation.
An exact factorisation combines spectral selection, conjugation, windowing, and reordering into a fixed map between one packed Fourier transform and the shorter band inverse transforms.
%Each band keeps its original length, while its spectral operations and inverse transform share a kernel.
The factors give waveform reconstruction, real adjoints for backpropagation, and bounds on arithmetic depth and block width; overlapping slices permit streaming with bounded memory.
%On NVIDIA A100 and V100 GPUs, 
Tests on two GPU models show that Flash-CQT reduces analysis–synthesis round-trip time by factors of two to eight relative to a baseline computing the same CQT.
The proposed implementation also uses over 30\% less peak temporary workspace and reaches a negligible reconstruction error, with a signal-to-noise ratio of about 130 dB, in single-precision floating-point arithmetic.
%The gains make the invertible CQT less costly for spectral processing and learning while preserving its resolution.
These advances make Flash-CQT a practical, computationally efficient front end for spectral analysis and modern audio machine-learning systems.
\end{abstract}


%\vspace{-3pt}
\section{Introduction}
%\vspace{-3pt}

Audio analysis often begins with a time-frequency representation~\cite{smithSASP,purwins2019deep}.
The widely used \ac{stft} analyses every frequency band with the same window, giving uniform resolution across the spectrum~\cite{allen1977stft,smithSASP}.
Equal musical intervals, however, correspond to equal frequency ratios, which motivates geometrically spaced analysis bands.
The \ac{cqt} uses this spacing and scales each window inversely with its centre frequency, keeping the ratio $Q$ of centre frequency to bandwidth constant~\cite{brown1991cqt}.
Long windows distinguish nearby low-frequency components, while short windows resolve rapid changes at high frequencies.
On the CQT frequency axis, equal pitch intervals span equal distances; its narrower low-frequency bands also reflect a broad feature of auditory frequency selectivity~\cite{glasberg1990erb}.
These properties make it useful for pitch estimation, transcription and music analysis~\cite{muller2015fmp,schorkhuber2010cqttoolbox}.

Although often used for analysis alone, the \ac{cqt} can also be made exactly invertible.
%so that its coefficients determine the waveform.
%The synthesis operation can be used to reconstruct, potentially edited or generated, coefficients to the waveform.
Its synthesis operation then maps coefficients, whether analysed, edited, or generated, back to a waveform.
The \ac{nsgt} provides the construction for such invertible transforms with frequency-dependent windows~\cite{balazs2011nsgt}.
Velasco \etal\ used it to obtain an exactly invertible \ac{cqt}~\cite{velasco2011icqt}, whereas Sch\"orkhuber and Klapuri achieved reconstruction by other means~\cite{schorkhuber2010cqttoolbox}; Holighaus \etal\ later derived a sliced form for streaming~\cite{holighaus2013slicq}.
Analysis and synthesis are linear
%over the reals 
and therefore differentiable, allowing gradients to pass through either operation during training.
In software, nnAudio provides differentiable \ac{cqt} analysis without an exact inverse~\cite{cheuk2020nnaudio}, while other implementations provide both invertibility and differentiation~\cite{moliner2022cqtpytorch,hanssian2021slicq}.
Their repeated use in training and inference nevertheless makes computational cost a concern.

The \ac{nsgt} implementation~\cite{balazs2011nsgt,velasco2011icqt} computes an invertible \ac{cqt} with one long \ac{fft}, selection and windowing of overlapping bands in the frequency domain, and a short inverse FFT of band-dependent length for each band.
%The Fourier filterbank approach~\cite{smith2009audio} evaluates an invertible \ac{cqt} with one long \ac{fft}, selection and windowing of overlapping bands, and unequal short inverse FFTs~\cite{velasco2011icqt}.
On GPUs, separate calls for these operations add kernel launches, temporary buffers and data movement between stages, while zero-padding bands to a common length wastes arithmetic and storage.
%Separate GPU calls for these operations introduce kernel launches, temporary buffers and repeated reads and writes of intermediate arrays.
%Zero-padding every band to a common length makes the arrays regular, but adds arithmetic and storage for otherwise unused entries.
%An operation count alone therefore misses the cost of moving data between the Fourier stages.

Trading arithmetic and layout against data movement is a long-standing theme in Fourier computation.
\ac{fft} frameworks organise matrix factors around permutations and memory access~\cite{cooley1965fft,van1992computational,gentleman1966fft,temperton1982mixed,johnson1990framework,bailey1990memory}.
Edelman \etal~\cite{edelman1999future} traded arithmetic for less communication in approximate distributed Fourier transforms.
Related concerns motivate GPU methods for structured products~\cite{dao2022monarch} and fused operations~\cite{fu2024flashfftconv,wu2025turbofno}; cuFFTDx permits FFTs inside custom CUDA kernels~\cite{nvidia2026cufftdx}.

%\begin{figure}[t]
%\centering
%%\resizebox{\columnwidth}{!}{
%\input{fig_flow_v2}
%%}
%\vspace{-10pt}
%\caption{For each schematic band $\lambda$, routing $R_\lambda$ and a length-$M_\lambda$ inverse FFT form one shaded kernel. It produces coefficients $c_\lambda$ from the shared packed spectrum $Z$.}
%\label{fig:flow}
%\end{figure}

\begin{figure}[t]
\centering
\resizebox{0.95\linewidth}{!}{
\input{fig_flow_final}
}
\vspace{-5pt}
%\caption{Three schematic bands indexed by $\lambda$. Packing $P$ and a length-$L$ FFT produce $Z$, where $L=N/2$. Each shaded kernel recovers and weights the band's bins through $R_\lambda$, then produces $c_\lambda$ with a length-$M_\lambda$ inverse FFT.}
\caption{Diagram of the proposed FlashCQT implementation, with three schematic bands indexed by $\lambda$.
Packing $P$ and a length-$L$ FFT give $Z$, with $L=N/2$. Each shaded kernel weights the band's bins via $R_\lambda$ and gives $c_\lambda$ by a length-$M_\lambda$ inverse FFT.}
\vspace{-5pt}
%\vspace{-pt}
\label{fig:flow}
\end{figure}
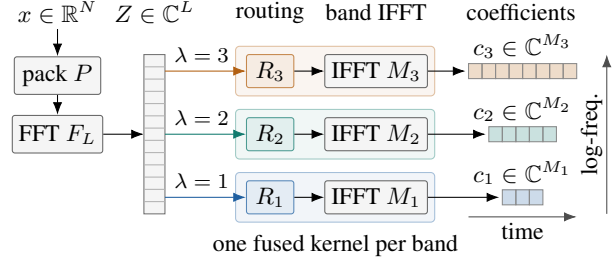

We propose a factorisation that combines the operations between the Fourier stages into a fixed map.
Pairing consecutive real samples as real and imaginary parts gives a complex sequence of half the length, whose FFT retains all the information needed to recover the signal's Fourier coefficients.
For each coefficient, a precomputed table identifies the FFT outputs to combine and includes the band's spectral window in their weights.
Theorem~\ref{thm:factorisation} bounds the depth and width of this map, which lets the routing and the band's short inverse FFT run in one GPU kernel, avoiding a write and read of the intermediate band spectrum (Fig.~\ref{fig:flow}).
Each band keeps its own length, and the factorisation gives reconstruction and real adjoints for backpropagation.
We call the resulting CUDA implementation Flash-CQT.

The rest of the paper is organised as follows.
Sec.~\ref{sec:transform} describes the invertible CQT in its default form.
Sec.~\ref{sec:structure} derives the packed factorisation, its depth and width bounds, and derives the adjoints for training.
Sec.~\ref{sec:implementation} details the CUDA implementation.
Sec.~\ref{sec:measurements} reports our evaluation, and  Sec.~\ref{sec:conclusion} concludes the paper.

\section{The invertible transform}
\label{sec:transform}
Let $F_n\in\C^{n\times n}$ be the unnormalised \ac{dft} matrix~\cite{van1992computational}, with row $k$ indexing frequency bins and column $j$ indexing samples, both from $0$ to $n-1$.
With $i^2=-1$ and $H$ denoting conjugate transpose, its entries and inverse are
\begin{equation}
 [F_n]_{kj}=e^{-2\pi i kj/n},\qquad F_n^{-1}=\tfrac{1}{n}F_n^H.
 \label{eq:dft}
\end{equation}
The factor $1/n$ normalises the inverse.
An FFT evaluates the product with $F_n$ without forming this matrix.
We use superscript $\top$ for transpose and an overbar for complex conjugation.

The \ac{cqt}, seen as a finite \ac{nsgt}, decomposes a real signal $x\in\R^N$ of even length $N$ into overlapping frequency bands indexed by $\lambda$, each carrying its own coefficient vector $c_\lambda\in\C^{M_\lambda}$ of length $M_\lambda$~\cite{balazs2011nsgt}.
In the constant-$Q$ case, bandwidth grows with frequency~\cite{velasco2011icqt}, so higher bands require more time coefficients over a fixed duration.
The resulting arrays are \emph{ragged}, with several distinct lengths $M_\lambda$.

Such transforms and their inverses can be defined through the \ac{dft} and computed with the \ac{fft}~\cite{velasco2011icqt,holighaus2013slicq}.
The analysis operator $T\colon x\mapsto\{c_\lambda\}_\lambda$ maps the input signal to the coefficients of all bands, where the operator $T_\lambda\colon x\mapsto c_\lambda$ of band $\lambda$ is
%The complete analysis operator $T: x \mapsto \{ c_\lambda \}_\lambda$ collects all bands, $Tx=\{T_\lambda x\}_\lambda$.
%The analysis operator $T_\lambda: x \mapsto c_\lambda$ of band $\lambda$ maps the input signal $x$ to the coefficients
\begin{align}
 c_\lambda&=T_\lambda x=F_{M_\lambda}^{-1}u_\lambda,
 \label{eq:analysis}\\
 u_\lambda&=E_\lambda\diag(g_\lambda)F_Nx.
 \label{eq:band-spectrum}
\end{align}
%\begin{equation}
% c_\lambda=T_\lambda x=F_{M_\lambda}^{-1}u_\lambda,
% \qquad u_\lambda=E_\lambda\diag(g_\lambda)F_Nx.
% \label{eq:analysis}
%\end{equation}
Analysis first computes the spectrum $X=F_Nx$ of the signal with a single \ac{dft} of length $N$.
Band $\lambda$ is specified by a real spectral window $g_\lambda\in\R^N$, applied by entrywise multiplication $\diag(g_\lambda)$, and a selection matrix $E_\lambda\in\R^{M_\lambda\times N}$.
The selection matrix places the bins where $g_\lambda$ is nonzero into their prescribed local frequency slots, leaving unused slots zero and sending no two bins to the same slot.
The resulting band weighted spectrum $u_\lambda$ is mapped by the short inverse \ac{dft} $F_{M_\lambda}^{-1}$ to the band's time coefficients $c_\lambda$.
Besides the positive-frequency bands, the transform retains two sidebands at zero frequency (DC) and half the sample rate (Nyquist), while negative-frequency bands follow by conjugate symmetry.

%Both the analysis and synthesis of such transforms can be defined with the \ac{dft} and computed with the \acp{fft}~\cite{velasco2011icqt,holighaus2013slicq}.
%We start with analysis. 
%Write $X=F_Nx$ for the Fourier coefficients of the input signal.
%%Write $X=F_Nx$ for the Fourier coefficients of the input signal.
%Band $\lambda$ is specified by a real spectral window $g_\lambda\in\R^N$ and a selection matrix $E_\lambda\in\R^{M_\lambda\times N}$, which places the bins where $g_\lambda$ is nonzero into their prescribed local frequency slots, leaving unused slots zero and sending no two bins to the same slot.
%Writing $\diag(g_\lambda)$ for entrywise multiplication by the window, the band analysis operator $T_\lambda$ is defined as 
%\begin{equation}
% c_\lambda = T_\lambda x = F_{M_\lambda}^{-1}u_\lambda,
% \qquad u_\lambda=E_\lambda\diag(g_\lambda)F_N x.
% \label{eq:analysis}
%\end{equation}
%Thus $u_\lambda$ is the band's weighted spectrum and $c_\lambda$ contains its time coefficients.
%We write $T$ for the analysis operator that maps $x$ to the collection of all $c_\lambda$.
%As bandwidth grows with frequency, higher bands need finer time sampling, which increases $M_\lambda$ in steps when bands are grouped.
%Besides the positive-frequency bands, the transform retains two sidebands at zero frequency (DC) and half the sample rate (Nyquist), while negative-frequency bands follow by conjugate symmetry.

We restrict throughout to the \emph{painless} case of frame theory~\cite{daubechies1986painless,balazs2011nsgt}, where band supports fit their local spectra without collisions and together cover every Fourier bin.
Reconstruction (synthesis) then requires only diagonal weighting.
%rather than solving a coupled linear system.
%To reconstruct the waveform, the overlapping windows must together cover each spectral bin.
Let $g_{\lambda,k}$ denote the $k$th entry of $g_\lambda$.
Because the windows overlap, bin $k$ receives contributions from several bands, which synthesis must compensate.
For $0\le k\le N/2$, the aggregate window weight at bin $k$ is $a_k=\sum_\lambda M_\lambda g_{\lambda,k}^2$, with the sum taken over the retained bands, and the dual weights are defined as the normalised windows
%Writing $g_{\lambda,k}$ for the window value at bin $k$, require positive coverage and define the dual weights $h_{\lambda,k}$ by
%\begin{equation}
$h_{\lambda,k}=M_\lambda g_{\lambda,k}/a_k$.
% \label{eq:weights}
%\end{equation}
%\begin{equation}
% a_k=\sum_\lambda M_\lambda g_{\lambda,k}^2>0,
% \qquad h_{\lambda,k}=(M_\lambda g_{\lambda,k})/a_k.
% \label{eq:weights}
%\end{equation}
The windows must cover every bin, so that $a_k>0$, and the factor $M_\lambda$ accounts for the normalisation of the band inverse \ac{dft} in~\eqref{eq:analysis}.
By construction, $\sum_\lambda h_{\lambda,k}g_{\lambda,k}=1$ at every bin.
%For $0\le k\le N/2$, the sum runs over the retained bands.
%These dual weights compensate for window overlap.
%These dual weights compensate for window overlap, while $E_\lambda^\top$ puts each band's bins back in the full spectrum.
%The factors $M_\lambda$ account for the inverse-DFT normalisation in~\eqref{eq:analysis}.
%Reconstruction requires only diagonal weights because the band supports fit without collisions.
%This is the painless construction~\cite{daubechies1986painless,balazs2011nsgt}.

%The synthesis operator $S$ maps each band back to its spectrum with a short \ac{dft}, returns the bins to their original positions, combines the bands and applies a single inverse \ac{dft},
%The inverse transform maps the coefficients to band spectra $u_\lambda=F_{M_\lambda}c_\lambda$, which coincide with~\eqref{eq:analysis} for analyzed coefficients, and combines them as
Synthesis reverses the steps of analysis.
The synthesis operator $S: \{c_\lambda \}_\lambda \mapsto \hat{x}$ maps the coefficients $c_\lambda$ of each band back to its spectrum with the short \ac{dft} $F_{M_\lambda}$, returns the bins to their original positions with $E_\lambda^\top$, weights them with the dual window $h_\lambda$, and applies the inverse \ac{dft} $F_N^{-1}$ to the sum over all bands,
\begin{equation}
 \hat x=S\bigl(\{c_\lambda\}_\lambda\bigr)=F_N^{-1}\sum_\lambda\diag(h_\lambda)E_\lambda^\top F_{M_\lambda}c_\lambda.
 \label{eq:synthesis}
\end{equation}
%\begin{equation}
% \hat x=S(c)=F_N^{-1}\sum_\lambda\diag(h_\lambda)E_\lambda^\top F_{M_\lambda}c_\lambda.
% \label{eq:synthesis}
%\end{equation}
As in analysis, the negative-frequency bands follow by conjugate symmetry, and the DC and Nyquist sidebands are made real, so that $\hat x$ is real.
%Taking real parts at DC and Nyquist, completing the negative frequencies by conjugation and applying $F_N^{-1}$ gives a real waveform.
For analysed coefficients, substitution of~\eqref{eq:analysis} gives $\hat x=x$, so $ST=I$ on real signals in exact arithmetic.
%For analyzed coefficients, substitution of~\eqref{eq:analysis} gives $\widehat x=x$ since $\sum_\lambda h_{\lambda,k}g_{\lambda,k}=1$.
%Write $S$ for the synthesis operator formed by~\eqref{eq:synthesis} and these final steps. 
%In exact arithmetic, $ST=I$ on real signals.
%Writing $S$ for these reconstruction steps gives $ST=I$ on real signals in exact arithmetic.

%For analyzed coefficients, $F_{M_\lambda}c_\lambda=u_\lambda$ as in~\eqref{eq:analysis}, and $\sum_\lambda h_{\lambda,k}g_{\lambda,k}=1$ gives $ST=I$ on real signals in exact arithmetic.

\section{Packed routing and arithmetic structure}
\label{sec:structure}

Because $x$ is real, its spectrum $X=F_Nx$ satisfies $X_{N-k}=\overline{X_k}$ for $1\le k<N$, so the bins $0\le k\le N/2$ determine the whole spectrum.
A full complex FFT does not exploit this redundancy.
%it processes $N$ complex values, $2N$ real numbers, to represent a signal of only $N$ real samples.
A classical real-input construction~\cite{sorensen1987realfft,temperton1983real} halves the transform length.
Let $L=N/2$ and pack consecutive real samples into complex pairs for $0\le j<L$,
\begin{equation}
 z_j=x_{2j}+ix_{2j+1}.
 \label{eq:pack}
\end{equation}
The packed spectrum $Z=F_Lz$ holds exactly $N$ real numbers, as many as the signal, and loses no information.
Any required bin of $X$ follows from two packed bins~\cite{van1992computational},
\begin{equation}
 X_k=\tfrac12(1-i\omega_k)Z_{k\bmod L}
       +\tfrac12(1+i\omega_k)\overline{Z_{-k\bmod L}},
 \label{eq:untangle}
\end{equation}
where $\omega_k=e^{-2\pi i k/N}$ and packed indices wrap modulo $L$.

Substituting~\eqref{eq:untangle} into~\eqref{eq:band-spectrum} gives each band spectrum directly from $Z$.
If entry $m$ of $u_\lambda$ holds bin $k$,
\begin{equation}
 [u_\lambda]_m=\tfrac12 g_{\lambda,k}\bigl[(1-i\omega_k)Z_{k\bmod L}+(1+i\omega_k)\overline{Z_{-k\bmod L}}\bigr],
 \label{eq:route}
\end{equation}
%with $\alpha_{\lambda,m}=\tfrac12 g_{\lambda,k}(1-i\omega_k)$ and $\beta_{\lambda,m}=\tfrac12 g_{\lambda,k}(1+i\omega_k)$; 
entries $m$ that hold no bin $k$ are zero.
We call the map $R_\lambda\colon Z\mapsto u_\lambda$ the \emph{router} of band $\lambda$.
It combines bin recovery, windowing and selection in one table without forming $X$.
Conjugation makes the router real-linear, though generally not complex-linear.
With $P$ denoting the packing in~\eqref{eq:pack}, the analysis of band $\lambda$ factors exactly as
\begin{equation}
 c_\lambda= T_\lambda x=F_{M_\lambda}^{-1}R_\lambda F_LP x.
 \label{eq:packed-analysis}
\end{equation}

Synthesis runs in the opposite direction.
Inverting~\eqref{eq:untangle} gives each packed bin from two bins of the half-spectrum,
\begin{equation}
 Z_q=\tfrac12
 \bigl[(1+i\overline{\omega_q})X_q+(1-i\overline{\omega_q})\overline{X_{L-q}}
 \bigr]
 ,
 \qquad 0\le q<L,
 \label{eq:retangle}
\end{equation}
where the conjugate term accounts for the negative frequencies.
In~\eqref{eq:synthesis}, band $\lambda$ contributes $X^{(\lambda)}_k=h_{\lambda,k}[E_\lambda^\top u_\lambda]_k$ for $0\le k\le N/2$, made real at DC and Nyquist.
Substituting $X^{(\lambda)}$ into~\eqref{eq:retangle} gives its packed spectrum $Z^{(\lambda)}$, and we call $\widetilde R_\lambda\colon u_\lambda\mapsto Z^{(\lambda)}$ the \emph{dual router} of band $\lambda$.
With $P^{-1}$ unpacking each complex sample into two real ones, synthesis is
\begin{equation}
 S\bigl(\{c_\lambda\}_\lambda\bigr)=P^{-1}F_L^{-1}\sum_\lambda\widetilde R_\lambda F_{M_\lambda}c_\lambda,
 \label{eq:packed-synthesis}
\end{equation}
%\begin{equation}
%\begin{align}
% S(\{ c_\lambda \}_\lambda)&=P^{-1}F_L^{-1}\sum_\lambda Z^{(\lambda)}, \\
%  \text{where}\quad Z^{(\lambda)}&=\widetilde R_\lambda F_{M_\lambda}c_\lambda,
% \label{eq:packed-synthesis}
%\end{align}
so the band contributions are added in the packed domain and a single inverse FFT of length $L$ serves all bands.

\subsection{Depth and block width}

Analysis~\eqref{eq:packed-analysis} has three stages, the packed transform $F_L$, routing $R_\lambda$, and band inverse transforms $F_{M_\lambda}^{-1}$; synthesis~\eqref{eq:packed-synthesis} reverses their order.
Each stage splits into \emph{layers} of small, independent linear blocks, which can be evaluated in parallel.
We describe this structure independently of hardware by its \emph{depth}, the number of layers that run in sequence, and its \emph{block width}, the larger of a block's real input and output dimensions.
Wiring between layers may reorder entries, change signs or insert zeros at no cost in this count, but may not duplicate inputs.

The Cooley--Tukey construction and its mixed-radix generalisations write a Fourier transform of length $n=r_1\cdots r_d$ as $d$ layers of small Fourier blocks, separated by permutations and twiddle factors~\cite{cooley1965fft,van1992computational,temperton1982mixed}.
The \emph{radix} $r_j$ is the length of a layer's Fourier blocks, so radix 2 combines pairs, radix 4 groups of four, and so on.
If every radix is at most $r$, let $\tau_r(n)$ be the least number of factors in a product $n=r_1\cdots r_d$ with $2\le r_j\le r$.
For example, $\tau_{32}(2^a)=\lceil a/5\rceil$, so the packed transform of length $32{,}768=32^3$ needs three layers.

%To describe the factorization independently of hardware, we count layers of independent small linear blocks, measuring each block's width by the larger of its real input and output dimensions.
%Between layers, wiring may permute entries, change signs or insert zeros, but does not duplicate inputs.
%If each small complex Fourier block has length at most $r$, let $\tau_r(n)$ be the least number of factors in a product $n=r_1\cdots r_d$ with $2\le r_j\le r$.
%For example, $\tau_{32}(2^a)=\lceil a/5\rceil$.

By~\eqref{eq:route}, every entry of $u_\lambda$ reads one pair of packed bins $\{Z_q,Z_{-q}\}$, namely when it holds bin $q$ or bin $L-q$.
Grouping entries by this unordered pair splits the router into independent blocks, each with at most four real inputs.
We call the number of nonzero entries, over all bands, that read the pair $\{Z_q,Z_{-q}\}$ its \emph{fan-out} $D_q$, and let $D=\max_q D_q$, so that no router block has more than $2D$ real outputs.

These counts give the following bound.

%Group the packed bins into unordered pairs $\{q,-q\bmod L\}$.
%Let $D_q$ count the nonzero band-spectrum entries supplied by a pair after folding the real spectrum, and let $D=\max_qD_q$.
%The count precedes the band inverse FFT.

\begin{samepage}
\begin{theorem}[Ragged Fourier factorisation]
\label{thm:factorisation}
If $L$ and every $M_\lambda$ factor into integers between $2$ and $r$, analysis and synthesis can be evaluated with depth at most
\begin{equation}
 d=\underbrace{\tau_r(L)}_{\text{packed FFT}}+\underbrace{1}_{\text{router}}+\underbrace{\max_\lambda\tau_r(M_\lambda)}_{\text{band FFTs}},
 \label{eq:depth}
\end{equation}
and real block width at most
\begin{equation}
 w=\max\bigl\{\underbrace{2r}_{\text{FFT block}},\underbrace{4}_{\text{router in}},\underbrace{2D}_{\text{router out}}\bigr\}.
 \label{eq:width}
\end{equation}
\end{theorem}
\end{samepage}

%\begin{samepage}
%\begin{theorem}[Ragged Fourier factorization]
%\label{thm:factorization}
%If $L$ and every $M_\lambda$ factor into integers between $2$ and $r$, analysis and synthesis can be evaluated with depth at most
%\begin{equation}
% d=\tau_r(L)+1+\max_\lambda\tau_r(M_\lambda),
% \label{eq:depth}
%\end{equation}
%and block width at most $\max\{2r,4,2D\}$.
%\end{theorem}
%\end{samepage}

%\begin{proof}
%Cooley--Tukey gives the first and last groups of layers, with twiddle factors absorbed into the small Fourier blocks~\cite{cooley1965fft,van1992computational}.
%After grouping by packed pair, the router consists of independent blocks with at most four real inputs and $2D$ real outputs.
%Both routers occupy one layer, since synthesis uses the same pairs with reversed dimensions.
%Identity layers align shorter band plans without zero-padding their coefficient vectors.
%\end{proof}
\begin{proof}
Cooley--Tukey gives the $\tau_r(L)$ and $\tau_r(M_\lambda)$ Fourier layers, with twiddle factors absorbed into the small Fourier blocks of at most $2r$ real values~\cite{cooley1965fft,van1992computational}.
Grouped by packed pair as above, the router consists of independent blocks with at most four real inputs and $2D$ real outputs.
The router and dual router each occupy one layer, since synthesis uses the same pairs with reversed dimensions.
Identity layers align shorter band plans without zero-padding their coefficient vectors.
\end{proof}

\subsection{Adjoints for training}
\label{sec:adjoints}
Backpropagation through analysis and synthesis uses their adjoints.
Write $S_\lambda$ for band $\lambda$'s contribution in~\eqref{eq:packed-synthesis}, so $S(c)=\sum_\lambda S_\lambda c_\lambda$.
Gradients use the real inner product $\langle u,v\rangle_{\R}=\Rea(u^Hv)$ on complex arrays, where $\Rea$ takes the real part.
Write $\dagger$ for its adjoint, which moves a linear map from one side of this inner product to the other.
Packing preserves the inner product, so $P^\dagger=P^{-1}$.
Because of its conjugate term, the router has the form $Rz=Az+B\overline z$ for complex matrices $A$ and $B$, and its real adjoint is $R^\dagger y=A^Hy+B^\top\overline y$.

Reversing the order of the factors and taking their real adjoints gives
\begin{align}
 T_\lambda^\dagger
   &=\frac{L}{M_\lambda}P^{-1}F_L^{-1}R_\lambda^\dagger F_{M_\lambda},
   \label{eq:analysis-adjoint}\\
 S_\lambda^\dagger
   &=\frac{M_\lambda}{L}F_{M_\lambda}^{-1}\widetilde R_\lambda^\dagger F_LP.
   \label{eq:synthesis-adjoint}
\end{align}
Thus $T_\lambda^\dagger$ has the form of synthesis with $R_\lambda^\dagger$ in place of $\widetilde R_\lambda$, and $S_\lambda^\dagger$ the form of analysis with $\widetilde R_\lambda^\dagger$ in place of $R_\lambda$.
Hence $R_\lambda^\dagger$ reads the same pairs as the dual router, with $g_{\lambda,k}$ instead of $h_{\lambda,k}$ and twice the weight at DC and Nyquist, so backpropagation reuses the same index tables and Fourier routines.
Because these maps are fixed and linear, no transform activations need be saved.

\section{Implementation}
\label{sec:implementation}
The CUDA implementation\footnote{\url{https://github.com/cucuwritescode/Flash-CQT}} uses FP32 arithmetic.
Table~\ref{tab:radix} gives the principal configuration, for a slice of $N=65536$ and 44.1~kHz, which follows the multi-resolution layout of~\cite{dacosta2025mrcqtdiff}.
It combines three CQT resolutions, covering 3, 4, and 2 octaves at 8, 16, and 32 bands per octave.
%, respectively.
The 152 positive-frequency bands share one coefficient length per octave, obtained by rounding its largest window support up to a power of two.
Band centres are rounded to Fourier bins, with the highest centre adjusted towards Nyquist.
DC and Nyquist sidebands complete the coverage.

\begin{table}[t]
\caption{Frequency band layout at $N=65,536$ and 44.1 kHz.
%Centre frequencies are rounded to hertz.
DC and Nyquist sidebands have lengths 128 and 2048.
}
\label{tab:radix}
\smallskip
\centering
\small
\setlength{\tabcolsep}{4pt}
\begin{tabular}{crrrr}
\toprule
Octave & Bands & Centres (Hz) & $M_\lambda$ & Radix plan\\
\midrule 
1 &  8 & 43--81        &   32 & $32$\\
2 &  8 & 88--165       &   64 & $8\times8$\\
3 &  8 & 180--337      &  128 & $16\times8$\\
4 & 16 & 345--666      &  128 & $16\times8$\\
5 & 16 & 695--1342     &  256 & $16\times16$\\
6 & 16 & 1402--2708    &  512 & $32\times16$\\
7 & 16 & 2830--5464    & 1024 & $32\times32$\\
8 & 32 & 5512--10905   & 1024 & $32\times32$\\
9 & 32 & 11147--21810  & 2048 & $32\times8\times8$\\
\bottomrule
\end{tabular}
\end{table}

The fast kernels are specialised to a slice of $N=\,\,$65,536, for which the packed length $L=32,768=32^3$ gives three radix-32 layers with fixed memory offsets.
The band kernels read routing tables and the radix plans of Table~\ref{tab:radix} at run time, but also assume this packed length.
Packing is performed within the first layer, and the last writes interleaved real and imaginary values in the order in which the FFT computes them.
The routing tables index that order directly~\cite{temperton1982mixed}.

One CUDA thread block of 256 threads handles each band and batch element, using at most 16.25 KiB of on-chip shared memory, including the precomputed Fourier weights.
The index maps avoid write conflicts and permit stages to overwrite their inputs within one buffer.
Synthesis evaluates~\eqref{eq:packed-synthesis}, using band FFTs and atomic FP32 additions to accumulate dual-routed contributions before the packed inverse FFT.

Routing records are stored in the order in which the threads read them.
If entry $m$ of $u_\lambda$ holds bin $k$, its 16-byte record holds two 16-bit indices into $Z$, the window value $g=g_{\lambda,k}$, and the complex weight $\alpha=\tfrac12 g(1-i\omega_k)$ from~\eqref{eq:route}.
Recovering the other weight as $\tfrac12 g(1+i\omega_k)=g-\alpha$ reduces table size.

For this layout, the fan-out is $D=7$, attained at packed pairs 59 through 63.
Theorem~\ref{thm:factorisation} gives depth $3+1+3=7$ and real block width 64 when $r=32$.
Under this radix cap, seven layers are necessary when Fourier-stage wiring preserves complex entries and the construction retains the packed spectrum $Z$, the band spectra $u_\lambda$, and a separate router layer.
Indeed, $d$ radix-32 layers let an output depend on at most $32^d$ complex inputs.
Every output of the packed DFT depends on all 32,768 inputs, and every output of a length-2048 band inverse DFT depends on its 1,411 nonzero inputs.
Each exceeds $32^2$ and requires three layers; the router requires one more.
This bound concerns the stated arithmetic structure, not runtime optimality or GPU kernel launches.

%\subsection{Sliced streaming}
%\label{sec:slicing}
Signals shorter than the slice length $N=\,\,$65,536 samples require omitting the lowest bands or lowering $Q$, which shortens their windows at the cost of frequency selectivity, and would need new specialised kernels. For signals longer than $N=\,\,$65,536, we transform slices of $N$ samples independently, with 50\% overlap and a Tukey transition window~\cite{holighaus2013slicq}.
%Each slice uses the same filterbank and kernels, so memory stays bounded as recording duration grows.
%The shifted windows sum to one, so overlap-adding the reconstructed slices recovers the waveform.
%The wrapper retains two half-slices of state and has one slice of algorithmic latency, or 1486.1 ms for the principal configuration.
%Sliced coefficients need not equal whole-signal coefficients.

\section{EVALUATION}
\label{sec:measurements}
Measurements use the configuration in Table~\ref{tab:radix}, with $N=65536$, FP32 arithmetic, and inputs and tables resident on the GPU.
The devices are an NVIDIA A100-SXM4 with 80 GB and a V100-SXM2 with 32 GB, both using PyTorch 2.12.0 and CUDA 12.6.
CUDA events give median times over 50 runs after three warm-ups; the text quotes A100 results unless stated otherwise.

The \emph{CQT baseline} evaluates the same operator in PyTorch using cuFFT-backed calls.
It computes a real FFT, gathers and weights the bands using precomputed tables, and batches equal-length band transforms by octave.
The inverse weights and combines their spectra before one real inverse FFT.
Both methods retain the same windows, native coefficient lengths, and sidebands, without zero-padding to a common length.
Differences in time and memory therefore come from the implementation, not from the operator.

%This compares the complete implementations, including routing, fusion and data layout.

The STFT provides a reference for the cost of a common transform with nearly the same number of coefficients on the same device.
With a 2048-sample Hann window, hops of 512 samples, and centred framing, it produces 132,225 complex coefficients, within 0.6\% of the CQT's 132,992, while using uniform frequency resolution.
Eager measurements use \texttt{torch.stft} and \texttt{torch.istft}.
Under graph replay, the inverse is an overlap-add with precomputed window normalisation, which agrees with \texttt{torch.istft} above 120 dB.

\begin{figure}[t]
\centering
\input{fig_perf_final}
 \vspace{-15pt}
 \caption{Time against batch size on A100 (solid) and V100 (dashed). (a) Round trip (analysis and synthesis) under CUDA graph replay. (b) Round trip with backward pass (analysis, synthesis, squared loss and backpropagation) in eager mode.}
%\caption{Round-trip (a) and differentiable-step (b) times on A100 (solid) and V100 (dashed). %Graph replay uses the overlap-add STFT inverse; the eager step uses \texttt{torch.istft}.
%}
\label{fig:perf}
\end{figure}
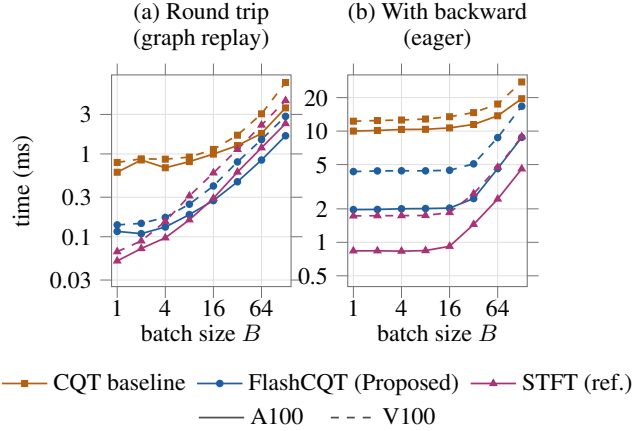

\subsection{Analysis and reconstruction}
Each round trip is captured as a CUDA graph, in Fig.~\ref{fig:perf}, which records its kernel sequence for inputs of fixed shape, and timed under replay.
This avoids Python dispatch between kernels and measures the GPU execution that the factorisation targets.
Flash-CQT is 5.2 times faster than the CQT baseline at $B=1$, and across batch sizes $B$ the speedup ranges from 2.1 to nearly 8 on A100 and 2.0 to 6.0 on V100 (Fig.~\ref{fig:perf}(a)).
Under graph replay, the proposed CQT takes 2.3 times the STFT time at $B=1$ and reaches parity near $B=16$.
This comparison depends on the inverse implementation, since at $B=128$ the native eager STFT round trip takes 1.39 ms, below the proposed CQT's 1.66 ms and the graph STFT's 2.32 ms.
%At $B=128$, FlashCQT takes 1.66 ms, against 2.32 ms for the STFT.

Table~\ref{tab:snr-mem} reports reconstruction \ac{snr}, defined as $10\log_{10}(\norm{x}_2^2/\norm{x-STx}_2^2)$, as the median over batch sizes.
Flash-CQT gives 129.8 dB, varying by at most 0.1 dB, or about $3\times10^{-7}$ relative error in the Euclidean norm.
The measured error is a few times the FP32 machine epsilon of $1.2\times10^{-7}$ and is consistent with round-off; the exact-arithmetic identity $ST=I$ does not imply bitwise reconstruction.
Transient workspace is the peak extra allocation during a call, excluding inputs and persistent tables, reported per signal because it grows linearly with $B$.
Keeping the spectrum packed and overwriting buffers in place reduces this workspace by 36 to 38\% across batches.

\begin{table}[t]
\vspace{-6.2pt}
\caption{Reconstruction SNR, workspace per signal and $B{=}1$ timings on an
NVIDIA A100. Round trips use graph replay; steps include eager backpropagation. Best values among the CQT implementations are bold.}
\label{tab:snr-mem}
\smallskip
\centering
\small
\setlength{\tabcolsep}{3pt}
\resizebox{\columnwidth}{!}{%
\begin{tabular}{lcrrrr}
\toprule
 & Multi- & SNR & Workspace & Round trip & Step\\
Method & res. & (dB) & (MB) & (ms) & (ms)\\
\midrule
STFT (reference)      & \xmark & 133.1 & 4.3 & 0.051 &  0.838\\
\midrule 
CQT baseline          & \cmark & 126.8 & 4.2 & 0.601 & 10.003\\
Flash-CQT (proposed)   & \cmark & \textbf{129.8} & \textbf{2.6} & \textbf{0.116} &  \textbf{1.969}\\
\bottomrule
\end{tabular}
}
\end{table}

\subsection{Differentiable round trip}
The second experiment adds squared error against a target and backpropagation through both transforms to measure a differentiable round trip without a neural network.
It runs in PyTorch eager mode, so its times include host dispatch and autograd overhead.
The baselines use PyTorch autograd, while Flash-CQT uses the adjoints in Section~\ref{sec:adjoints}.
Analysis and synthesis gradients agree with autograd at 130.05 and 129.48 dB, respectively.

Eager speedups over the CQT baseline range from 2.2 to 5.3 on A100 and 1.7 to 3.0 on V100 (Fig.~\ref{fig:perf}(b)).
In the baseline, overlapping windows make the backward pass of the band gather a scatter-add with repeated indices, whereas the explicit adjoints reuse the Fourier stages and routing tables.
Relative to the STFT, the CQT baseline costs 12 times as much at $B=1$ and 4.3 times at $B=128$, whereas FlashCQT costs 2.35 and 1.9 times as much.
Effects on complete model training and comparison with a matched cuFFTDx implementation remain open.

\section{Conclusion}
\label{sec:conclusion}
%\textcolor{red}{TODO: it is important revise and extend this. We must remark the potential benefits and applications of this implementation}

This paper presented Flash-CQT, an exact factorisation of the invertible \ac{cqt} for efficient GPU computation.
This factorisation combines the operations between the Fourier stages without changing the invertible \ac{cqt}.
Fixed routing recovers and weights each band's bins inside its transform kernel, preserving native lengths, while the factors give reconstruction, real adjoints, and arithmetic bounds.
On the tested GPUs, this organisation reduces round-trip time by factors of two to eight against the CQT baseline, and uses over a third less temporary workspace while retaining about 130 dB reconstruction \ac{snr}.

%The construction also applies to other painless NSGTs with real spectral windows satisfying the stated support and radix conditions, beyond the CQT configuration evaluated in this paper.

%Future work will also measure the gains in complete model training and compare against a similarly fused \ac{cqt} implementation.

%Their routing degrees and performance remain to be established, as do the gains in complete model training and against a comparably fused CQT baseline.
%Together these questions extend the present work towards efficient invertible representations with more general frequency resolutions.

These savings matter most where the transform is applied repeatedly, as in the training and inference of neural audio models.
Applications that require invertible \ac{cqt} coefficients, such as diffusion-based audio generation~\cite{dacosta2025mrcqtdiff}, music source separation~\cite{hanssian2021slicq}, audio bandwidth extension~\cite{moliner2024blind,ali2025exploiting} and timbre transfer~\cite{huang2019timbretron}, can therefore use the transform within a small factor of the STFT cost while retaining its logarithmic frequency resolution.
\ac{cqt}-based training losses and discriminators~\cite{li2024mert,welker2025flowdec,gu2024cqtdiscriminator} can likewise benefit from the fused analysis kernels even without invertibility, and from the exact adjoints of Sec.~\ref{sec:adjoints} for their gradients.
The construction also applies to other painless \acp{nsgt} beyond the CQT configuration evaluated in this paper, such as auditory filterbanks~\cite{necciari2013erblet}, although their fan-out and performance remain to be established.
Slicing further extends the implementation to long recordings and streaming.

%applications
%Audio generation \cite{dacosta2025mrcqtdiff}
%Bandwidth extension \cite{moliner2024blind,ali2025exploiting}
%Source separation \cite{hanssian2021slicq}
%timbre transfer \cite{huang2019timbretron}

%losses
%discriminators \cite{gu2024cqtdiscriminator}
%losses for audio coding \cite{welker2025flowdec}
%representation learning  \cite{li2024mert}

\clearpage
\balance
%\section{Compliance with Ethical Standards}
%This work studies numerical algorithms and their GPU implementation. It involved no human or animal subjects.

\section{Acknowledgements}
\vspace{-5pt}
We acknowledge the computational resources provided by the Aalto Science-IT project. This work was supported by the HUCE infrastructure of the Aalto School of Electrical Engineering. This research was funded in part by the Research Council of Finland (grant no. 371845).

%{\small
\renewcommand{\thebibliography}[1]{%
  \section{References}\list
  {[\arabic{enumi}]}{\settowidth\labelwidth{[#1]}\leftmargin\labelwidth
   \advance\leftmargin\labelsep
   \usecounter{enumi}%
   \setlength{\itemsep}{1pt plus 0.2ex}%
   \setlength{\parsep}{0pt plus 0.2ex}%
   \setlength{\topsep}{0pt plus 0.2ex}}%
  \def\newblock{\hskip .11em plus .33em minus .07em}
  \sloppy\clubpenalty4000\widowpenalty4000
  \sfcode`\.=1000\relax}
\let\endthebibliography=\endlist
\bibliographystyle{IEEEbib}
\bibliography{refs}
%}
\end{document}

%% file: fig_flow_final.tex
\begin{tikzpicture}[x=1cm,y=1cm,>=Latex,font=\fontsize{9}{10.5}\selectfont]
\definecolor{bandA}{RGB}{31,95,160}
\definecolor{bandB}{RGB}{17,122,110}
\definecolor{bandC}{RGB}{178,96,18}
\tikzset{
  blk/.style={draw=black!60,fill=black!4,align=center,inner sep=2.5pt,
              rounded corners=1pt},
  txt/.style={inner sep=1pt,align=center}}
\node[txt] at (.65,1.62) {$x\in\mathbb{R}^{N}$};
\node[blk,minimum width=1.15cm,minimum height=.50cm] (P) at (.65,.78) {pack $P$};
\node[blk,minimum width=1.15cm,minimum height=.50cm] (FL) at (.65,0) {FFT $F_L$};
\draw[->] (.65,1.40)--(P.north);
\draw[->] (P.south)--(FL.north);
\draw[->] (FL.east)--(1.80,0);
\fill[black!4] (1.80,-1.06) rectangle (2.08,1.06);
\draw[black!55] (1.80,-1.06) rectangle (2.08,1.06);
\foreach \i in {1,...,12}{
  \draw[black!30] (1.80,{-1.06+2.12*\i/13})--(2.08,{-1.06+2.12*\i/13});}
\node[txt] at (1.94,1.62) {$Z\in\mathbb{C}^{L}$};
\foreach \lab/\yy/\cnt/\col in {1/-.84/3/bandA,2/0/5/bandB,3/.84/8/bandC}{
  \draw[\col!40,fill=\col!6,rounded corners=1.5pt]
    (3.02,\yy-.33) rectangle (5.67,\yy+.33);
  \node[blk,draw=\col!75,fill=\col!12,minimum width=.63cm,
        minimum height=.44cm] (r\lab) at (3.48,\yy) {$R_{\lab}$};
  \node[blk,minimum width=1.36cm,minimum height=.44cm]
        (f\lab) at (4.89,\yy) {IFFT $M_{\lab}$};
  \draw[\col,->] (2.08,\yy)--(r\lab.west);
  \node[txt] at (2.57,\yy+.22) {\footnotesize $\lambda=\lab$};
  \draw[->] (r\lab.east)--(f\lab.west);
  \pgfmathsetmacro{\xs}{6.84-.09*\cnt}
  \draw[->] (f\lab.east)--(\xs-.04,\yy);
  \foreach \j in {1,...,\cnt}{
    \fill[\col,opacity=.22] ({\xs+.18*(\j-1)},\yy-.10) rectangle ++(.18,.20);
    \draw[black!45] ({\xs+.18*(\j-1)},\yy-.10) rectangle ++(.18,.20);}
  \node[txt] at (6.84,\yy+.33) {$c_{\lab} \in \mathbb{C}^{M_\lab}$};
}
\draw[black!70,->] (6.12,-1.10)--(7.56,-1.10);
\node[txt,anchor=north] at (6.84,-1.14) {time};
\node[txt,rotate=90] at (7.76,0) {log-freq.};
\draw[black!70,->] (7.96,-1.06)--(7.96,1.06);
\node[txt] at (3.48,1.62) {routing};
\node[txt] at (4.89,1.62) {band IFFT};
\node[txt] at (6.82,1.62) {coefficients};
\node[txt,anchor=north] at (4.35,-1.30) {one fused kernel per band};
\end{tikzpicture}

%% file: fig_perf_final.tex
\begin{tikzpicture}
\definecolor{cProp}{RGB}{31,95,160}   
\definecolor{cCuff}{RGB}{178,96,18}   
\definecolor{cStft}{RGB}{170,51,119}  
\pgfplotsset{
  every axis/.append style={
    font=\fontsize{9}{10.5}\selectfont,
    xmode=log, log basis x=2, ymode=log,
    xtick={1,4,16,64}, xticklabels={1,4,16,64},
    xmin=0.85, xmax=150,
    log ticks with fixed point,
    grid=major,
    grid style={black!12,line width=.3pt},
    axis line style={black!55},
    tick align=outside, tick style={black!55},
    xlabel near ticks, ylabel near ticks,
    width=.46\columnwidth, height=.51\columnwidth,
    xlabel={batch size $B$},
    xticklabel style={yshift=2.5pt},
    xlabel style={yshift=3pt},
    title style={align=center},
  },
  every axis plot/.append style={line width=.65pt},
}

\begin{groupplot}[group style={group size=2 by 1, horizontal sep=.75cm}]

\nextgroupplot[
  title={(a) Round trip\\(graph replay)},
  ylabel={time (ms)},
  ymin=0.025, ymax=9,
  ytick={0.03,0.1,0.3,1,3}, yticklabels={0.03,0.1,0.3,1,3},
  legend style={draw=none, fill=none, font=\fontsize{9}{10.5}\selectfont, legend columns=3,
    /tikz/every even column/.append style={column sep=2pt}},
  legend to name=legperf,
]
\addlegendimage{cCuff,solid,mark=square*,mark size=1.0pt}\addlegendentry{CQT baseline}
\addlegendimage{cProp,solid,mark=*,mark size=1.1pt}\addlegendentry{FlashCQT (Proposed)}
\addlegendimage{cStft,solid,mark=triangle*,mark size=1.3pt}\addlegendentry{STFT (ref.)}
\addplot[cCuff,solid,mark=square*,mark size=1.0pt,mark options={solid}] coordinates {(1,0.601) (2,0.845) (4,0.681) (8,0.808) (16,0.996) (32,1.263) (64,1.769) (128,3.599)};
\addplot[cProp,solid,mark=*,mark size=1.1pt,mark options={solid}] coordinates {(1,0.116) (2,0.109) (4,0.131) (8,0.185) (16,0.273) (32,0.461) (64,0.844) (128,1.656)};
\addplot[cStft,solid,mark=triangle*,mark size=1.3pt,mark options={solid}] coordinates {(1,0.051) (2,0.072) (4,0.097) (8,0.160) (16,0.291) (32,0.604) (64,1.186) (128,2.324)};
\addplot[cCuff,dashed,mark=square*,mark size=1.0pt,mark options={solid}] coordinates {(1,0.789) (2,0.872) (4,0.864) (8,0.916) (16,1.133) (32,1.684) (64,3.049) (128,7.29)};
\addplot[cProp,dashed,mark=*,mark size=1.1pt,mark options={solid}] coordinates {(1,0.139) (2,0.145) (4,0.171) (8,0.247) (16,0.41) (32,0.801) (64,1.49) (128,2.842)};
\addplot[cStft,dashed,mark=triangle*,mark size=1.3pt,mark options={solid}] coordinates {(1,0.066) (2,0.089) (4,0.153) (8,0.311) (16,0.59) (32,1.133) (64,2.229) (128,4.41)};

\nextgroupplot[
  title={(b) With backward\\(eager)},
  ymin=0.4, ymax=32,
  ytick={0.5,1,2,5,10,20}, yticklabels={0.5,1,2,5,10,20},
  legend style={draw=none, fill=none, font=\fontsize{9}{10.5}\selectfont, legend columns=2,
    /tikz/every even column/.append style={column sep=8pt}},
  legend to name=leggpu,
]
\addlegendimage{black!70,solid}\addlegendentry{A100}
\addlegendimage{black!70,dashed}\addlegendentry{V100}
\addplot[cCuff,solid,mark=square*,mark size=1.0pt,mark options={solid}] coordinates {(1,10.003) (2,10.131) (4,10.343) (8,10.373) (16,10.684) (32,11.470) (64,13.728) (128,19.556)};
\addplot[cProp,solid,mark=*,mark size=1.1pt,mark options={solid}] coordinates {(1,1.969) (2,1.975) (4,2.003) (8,2.012) (16,2.031) (32,2.474) (64,4.601) (128,8.779)};
\addplot[cStft,solid,mark=triangle*,mark size=1.3pt,mark options={solid}] coordinates {(1,0.838) (2,0.839) (4,0.833) (8,0.843) (16,0.921) (32,1.449) (64,2.448) (128,4.573)};
\addplot[cCuff,dashed,mark=square*,mark size=1.0pt,mark options={solid}] coordinates {(1,12.272) (2,12.444) (4,12.575) (8,12.836) (16,13.463) (32,14.663) (64,17.508) (128,27.65)};
\addplot[cProp,dashed,mark=*,mark size=1.1pt,mark options={solid}] coordinates {(1,4.327) (2,4.386) (4,4.397) (8,4.393) (16,4.445) (32,5.076) (64,8.734) (128,16.713)};
\addplot[cStft,dashed,mark=triangle*,mark size=1.3pt,mark options={solid}] coordinates {(1,1.727) (2,1.735) (4,1.737) (8,1.744) (16,1.852) (32,2.735) (64,4.727) (128,8.913)};
\end{groupplot}

\coordinate (mid) at ($(group c1r1.south)!0.5!(group c2r1.south)$);
\node[anchor=north] (lm) at ([yshift=-2pt]mid |- group c1r1.outer south) {\pgfplotslegendfromname{legperf}};
\node[anchor=north] at ([yshift=9pt]lm.south) {\pgfplotslegendfromname{leggpu}};
\end{tikzpicture}